\documentclass[12pt]{article}
\usepackage{bbm}
\usepackage[english]{babel}
\usepackage[latin1]{inputenc}
\usepackage{booktabs}
\usepackage{graphicx}
\usepackage{graphics}
\usepackage{caption}
\usepackage{amsmath}

\usepackage{amssymb}
\usepackage{amsthm}
\theoremstyle{plain}
\newtheorem{theorem}{Theorem}[section]
\newtheorem{lemma}[theorem]{Lemma}
\newtheorem{proposition}[theorem]{Proposition}

\theoremstyle{definition}
\newtheorem{definition}{Definition}[section]
\theoremstyle{remark}

\usepackage{hyperref}
\usepackage{array}
\usepackage{color}
\usepackage[]{tabularx}
\usepackage{float}

\title{Risk Measures Estimation Under Wasserstein Barycenter}
\author{
	\textbf{M. ANDREA ARIAS-SERNA} \\
	{\normalsize Faculty of Engineering} \\
	{\normalsize University of Medellin, Colombia} \\
	{\normalsize marias@udem.edu.co} \\
	\\
	\textbf{JEAN-MICHEL LOUBES} \\
	{\normalsize Toulouse Mathematics Institute} \\
	{\normalsize University Paul Sabatier, France} \\
	{\normalsize loubes@math.univ-toulouse.fr}\\
\\
\textbf{FRANCISCO J. CARO-LOPERA} \\
	{\normalsize Faculty of Basic Sciences} \\
	{\normalsize University of Medellin, Colombia} \\
	{\normalsize fjcaro@udem.edu.co} \\
}

\date{}

\begin{document}

\maketitle

\begin{abstract}

Randomness in financial markets requires modern and robust multivariate models of risk measures.  This paper proposes a new approach for modeling multivariate risk measures under Wasserstein barycenters of probability measures supported on location-scatter families. Simple and advanced copulas multivariate Value at Risk models are compared with the derived technique. The performance of the model is also checked in market indices of United States  generated by the financial crisis due to COVID-19. The introduced model behaves satisfactory in both common and volatile periods of asset prices, providing realistic VaR forecast in this era of social distancing.

\textbf{keywords}
Wasserstein Barycenter; Value at Risk; Conditional Value-at-risk; Location-scatter family; Transportation Cost.

\end{abstract}

\section{Intoduction} \label{agg}

When the well known univariate risk measure analysis is generalized into the multivariate setting, a number of complex theoretical and applied problems appear. This emerging theory just extends the univariate case for VaR estimate by using copulas, theory of extreme value, Monte Carlo method, historical simulation, variance and covariance analysis, historical simulation, among many others. However, some of the univariate translations became unrealistic and are based on inappropriate assumptions. For example, under the restriction of a perfect dependence, the simple summation method computes the total risk by summation of the stand alone risks; a preservative method with statical benefit (see for example, \cite{Embrecht(2013)} and \cite{Li(2015)}). In a similar way, for a large number of assets, the variance covariance method fails, because the estimation of the corresponding matrix is extremely cumbersome  due to the high amount of correlations, see \cite{McNeil(2015)}.

Now, in the context of risk management, the multivariate theory of the extreme value (EVT) (see \cite{McNeil(1999)}) and the popular technique of multivariate copulas (see \cite{Embrecht(2002)}) are useful in some scenarios for VaR estimation in portfolios.  In particular, copulas method attains a robust structure for dependence in financial time series by producing joint distributions with known non gaussian marginal distributions. Modelling the marginal distributions via copulas allows VaR computations with a better performance than the classical methods; but it involves some untractable assumptions in the context of risk measures which are difficult to elucidate; a similar quotation for the multivariate extreme value theory are also addressed by \cite{Jin(2018)} and \cite{Barone(2015)}.

Inspired by the above discussions and some interesting insights collected by \cite{Jianping(2012)} for risk models in banking industry,  this work proposes risk measures based on the Wasserstein barycenter. The new method considers a reliable risk measure based on distances among probabilistic models. The underlying suitable probability laws obey for example, opinions, beliefs and estimates of data sources, in the context of the financial risk. Explicitly, a far away concept in probability theory is brought into the financial models by proposing the named Fréchet risk measures; which are calibrated by certain metrization of the probability measure space. In this case, the well studied metric of Wasserstein supports the method and provides fundamental connections for the rising concept of barycenter in the sense of Agueh and Carlier in [1]. A seminal work for a number of generalizations and applications, see for example, \cite{Bigot(2018)}, \cite{Alvarez(2018)}, \cite{Loubes(2017)} and the references therein. A crucial aspect underlies here for the new method: the proposed barycenter remains invariant under a class of location and scatter set of (finite or infinite) set of probabilities.

The Wasserstein metric has enriched notably the risk management literature, see for example \cite{Kiesel(2016)} and \cite{Feng(2018)}. In particular, as a canonical metric under well defined assumptions, the named robust risk management has been studied under the Wasserstein metric. As a consequence, this work represents risk measures via statistical functionals by hybridizing robustness and continuity under the Wasserstein metric. Then, several financial applications of the Wasserstein metric can be obtained in real time series, where the classical methods provide an excellent scenario for the correctness of the predictions.

The above discussion is organized in this paper as follows: preliminaries about Wasserstein barycenter are given in Section 2, in order to establish the existence and uniqueness theorem for barycenter of distributions under a class of location and scatter  distributions. Then Section 3 defines the Wasserstein Barycenter in risk measures and results for VaR and CVaR estimation are given under the addressed family of scale and location distributions. Finally, Section 4 applies the the results in a portfolio consisting of two assets (the Nasdaq and the S\&P500 stock indices).

\section{Barycenters in the Wasserstein space: General results}

This apart provides the necessary background about the barycenter in a Wasserstein space. First the quadratic transportation cost is presented, in order to introduce the 2-Wasserstein distance. Then, the basic multivariate Value At Risk models are revisited in terms of the  Wasserstein barycenter. Finally, following \cite{Alvarez(2018)} and \cite{Villani(2008)}, the fundamental results on Wasserstein barycenters for measures under a class of scatter and location distributions are provided.

Start with $\mathcal{P}_{2}(\mathbb{R}^d)$ as the set of all probability measures defined on $\mathbb{R}^d$ with a finite second order moment. Denote $\mathcal{P}_{2,ac}(\mathbb{R}^d)$ as the subset of absolutely continuous measures and consider $(\Omega, \sigma,P)$ as a generic probability space. If  $\mu,\nu$ in $P(\mathbb{R}^d)$, are two measures, then  $\mathcal{P}(\mu,\nu)$ will denote the set of all probability measures $\pi$ in the product set $\mathbb{R}^d\times \mathbb{R}^d$. Here,  $\mu$ and $\nu$ are the corresponding first and second marginals.

Now, for two measures $\mu, \nu$ in $\mathcal{P}(\mathbb{R}^d)$, \textbf{the quadratic transportation cost} between  $\mu$ and $\nu$  (also referred as the transportation cost with a quadratics cost function) is defined as follows
$$\mathcal{T}_{2}(\mu,v)=inf_{\pi \in \mathcal{P}(\mu,v)}\int_{\mathbb{R^d}\times\mathbb{R^d}} d(x,y)^2d\pi(x,y).$$

The transportation cost with quadratics cost function endows the set $\mathcal{P}_{2}(\mathbb{R}^d)$ with the metric called \textbf{2-Wasserstein distance} or  Monge-Kantorovich distance metric, which is given by
$$W_{2}(\mu,\upsilon)=\mathcal{T}_{2}(\mu,\nu)^{\frac{1}{2}}.$$

When $d=1$, the 2-Wasserstein distance in the real line is just given by the quantile-like expression:

$$W_{2}^{2}(\mu,\upsilon)=\int_{0}^{1}|F_{\nu}^{-1}(x)-F_{\mu}^{-1}(x)|^2dx,$$

here $F_{\nu}^{-1}$  and  $F_{\mu}^{-1}$) are the quantile function of $\nu$ and  $\mu$, respectively.

Now, in the Euclideam space, the barycenter of points $x_{1},...x_{N}$ with weights $\lambda_{1},...,\lambda_{N}, \lambda_{j}\geq 0$, $\sum_{j=1}^{N}\lambda_{j}=1$, is defined as

$$b=\sum_{j=1}^{N}\lambda_{J}x_{j}.$$

In fact, is the unique minimizer
$$E(y)=\sum_{j=1}^{N}\lambda_{j}|x_{j}-y|^{2}$$

Motivated by the Euclidean version, the Wasserstein barycenter can be defined as follows.

\begin{definition} Let $\mu_{1},...,\mu_{N}$ be  random probability measures over $\mathbb{R}^d$, endowed with positive weights $\lambda_{1},...\lambda_{N}$, with  $\sum_{j=1}^{N} \lambda_{j}=1$. The measure $\mu\in \mathcal{P}_{2}(\mathbb{R}^d)$ is a Wasserstein barycenter,  if $\mu$ is a minimizer of the functional
	\begin{equation}
	E(\mu)=\sum_{j=1}^{N}\lambda_{j}W_{2}^{2}(\mu;\mu_{j})
	\end{equation}
This fact will be denoted by
$\mu_{B}(\lambda)\in Bar((\mu_{j},\lambda_{j})_{1\leq j\leq N})$
	\end{definition}

\begin{definition} We say that the measure $\mu\in \mathcal{P}_{2}(\mathbb{R}^d)$ is a Wasserstein barycenter for the random probability measures $\mu_{1},...,\mu_{N}$ over $\mathbb{R}^d$, endowed with positive weights $\lambda_{1},...\lambda_{N}$, where  $\sum_{j=1}^{N} \lambda_{j}=1$, if $\mu$ is a minimizer of
	\begin{equation}
	E(\mu)=\sum_{j=1}^{N}\lambda_{j}W_{2}^{2}(\mu;\mu_{j})
	\end{equation}
	We will write
$\mu_{B}(\lambda)\in Bar((\mu_{j},\lambda_{j})_{1\leq j\leq N})$
	\end{definition}

Empirical consistency of the Wasserstein barycenter has been studied in \cite{Agueh(2011)}, \cite{Loubes(2015)} and \cite{Loubes(2017)}.

For introducing a fundamental result, we recall the  definition of a location-scatter class.

\begin{definition} If $M^{+}_{d \times d}$ denotes the set of ${d \times d}$ positive definite matrices and  $X_{0}$ is a random vector with measure $\mu_{0}\in\mathcal{P}_{2,ac}(\mathbb{R}^d)$, then the
	 set $\mathfrak{F}(\mu_{0})$ of probability laws defined by
$$ \mathfrak{F}(\mu_{0}):=\{\mathfrak{l}(AX_{0}+m): A\in M^{+}_{d \times d}, m\in \mathbb{R}^d \}$$
 is a location-scatter family induced by positive definite affine transformations from $\mu_{0}$.
\end{definition}

The Wasserstein barycenters of measures on a location-scatter family satisfies the following remarkable property, see \cite{Alvarez(2018)}.
\begin{lemma}[Theorem 3.10., \cite{Alvarez(2018)}]
Let $\mu_{0}\in \mathcal{P}_{2,ac}(\mathbb{R}^d)$, and $\mu \in W_2(P_2(\mathbb{R}^d))$, assume that for every $\omega\in \Omega$, the measure $\mu_{\omega}\in \mathfrak{F}(\mu_{0})$. Then the unique barycenter, $\overline{\mu}$ of $\mu$ also belongs to $\mathfrak{F}(\mu_{0})$. The mean of $\hat{\mu}$ is $\overline{m}:=\int m_{\omega}P(d\omega)$, and the covariance matrix, $\overline{\Sigma}$, is the only positive definite root of the equation
$$\overline{\Sigma}=\int(\overline{\Sigma}^{\frac{1}{2}}\Sigma_1{\omega}\overline{\Sigma}^{\frac{1}{2}})^{\frac{1}{2}}P(d\omega)$$
\end{lemma}

This result means that the Wasserstein barycenters are closed respect the location-scatter class.

An interesting case follows for $N$ Gaussian measures on $\mathbb{R}^{d}$:

\begin{lemma}[Theorem 2.5., \cite{Alvarez(2018)}]
	Consider $N$ Gaussian measures $\mu_{1},...,\mu_{N}$ on $\mathbb{R}^{d}$ with corresponding means $m_1,...,m_N$ and positive definite covariances $\Sigma_1,...,\Sigma_N$, and let $\lambda_{1},...,\lambda_{N}$ be positive weights with $\sum_{j=1}^{N} \lambda_{j}=1$. Then the unique barycenter of the normal measures $\mu_{1},...,\mu_{N}$ is the Gaussian distribution with mean $\overline{m_{\lambda}}=\sum_{j=1}^{N}\lambda_{j}m_{j}$  and  covariance matrix $\overline{\Sigma}$, which is the only positive definite root of the equation
	$$\overline{\Sigma}=\sum_{i=1}^{N}\lambda_{i}(\Sigma^{\frac{1}{2}}\Sigma_{i}\Sigma^{\frac{1}{2}})^{\frac{1}{2}}$$
\end{lemma}

According to \cite{Alvarez(2018)}, Wasserstein barycenters inherits the strong computational problems of the classical optimal transportation. However, in the real line some explicit distributions can be obtained.

\begin{proposition}\label{propo}
Let $F_{1}^{-1},...,F_{N}^{-1}$ be the quantile functions corresponding to $\mu_{1},...,\mu_{N}$ in the real line. Thus the barycenter of $\mu_{1},...,\mu_{N}$ is the probability with quantile function $\sum_{j=1}^{N} \lambda_{j}F_{i}^{-1},$
where  $\lambda_{1},...,\lambda_{N}$ are positive weights such that $\sum_{j=1}^{N} \lambda_{j}=1$.
\end{proposition}

Finally, using proposition \ref{propo}, with $N$ Gaussian distributions, $N(m_{i},\sigma_{i}^{2})$, $i=1,...,k$, on $\mathbb{R}$, then barycenter is Gaussian $N(\sum_{j=1}^{N}\lambda_{j}m_{j}, (\sum_{j=1}^{N}\lambda_{j}\sigma_{j})^2).$\\

This notable aspect will be used in the context of risk measures.

\section{Wasserstein Barycenter Risk Measures}\label{principal}

This section proposes the Wasserstein Barycenter Risk Measures at a confidence level $\alpha$. The research considers risk measures such as Value at Risk and Conditional Value at Risk for a loss random variable defined by $X^+=\sum_{i=1}^{N}\omega_{i}X_i.$ Here $X_1,...,X_N$ are loss random variables attributed to risk types endowed with positive weights $\omega_{1},...\omega_{N}$ (such that  $\sum_{j=1}^{N} \omega_{j}=1$) and over a fixed time period T. Now, for computation of $VaR_{\alpha}(X^+)$,  a joint law for the random vector $(X_1, . . . , X_N)^{'}$ is required. The Wasserstein barycenter can be regarded as the aggregate model for certain set of probability measures. It is also suitable for reaching an "average" distribution. The procedure also considers and optimal selection for the positive weights. They are connected with the source credibility for every prior. Moreover, the weights must be chosen equal when all priors remains acceptable. The equality also holds under unknowing performance reliability of the competing laws.

We are in position for definition of the Wasserstein Barycenter Value-at-Risk.

\begin{definition}\label{defq}
Given the aggregate position $X^{+}$,  a set of measures $M=(\mu_{1},...,\mu_{N})$, a set of weights $\lambda=(\lambda_{1},...,\lambda_{N})\in \mathbb{R}^{N-1}$ and a set of quantiles  $F=(F_{\mu_{1}}^{-1},...,F_{\mu_{N}}^{-1} )$ with $\alpha\in(0,1)$. The Wasserstein Barycenter Value-at-Risk is defined as:
\begin{equation}
VaR_{\alpha}(X^{+},\lambda)=F^{-1}_{\mu_{B}(\lambda)}(\alpha)
\end{equation}
Where $F^{-1}_{\mu_{B}(\lambda)}$ is the quantile function of Wasserstein barycenter $\mu_{B}(\lambda)$ of  $\mu_{1},...,\mu_{N}$ with weights $\lambda_{1},...\lambda_{N}\in \mathbb{R}$, where $\lambda_{j}\geq 0$, $1\leq j\leq N$, $\sum_{j=1}^{N} \lambda_{j}=1$.
\end{definition}

\subsection{Wasserstein Barycenter risk measures for location and scale distributions}

Next we use the notable property that the barycenter of distributions of location and scale families belongs to the same class. This allows to derive a closed-form formulas for the  Wasserstein Barycenter risk measures for location and scale distributions. \\

\begin{theorem}\label{tV} Let $\mu_{1},...,\mu_{N}$ be location and scale measures with corresponding means $m_1,...,m_N$  and standard deviations $\sigma_1,...,\sigma_N$; then the  Wasserstein Barycenter Value-at-Risk ($VaR_{\alpha}(X^{+},\lambda)$) is given by
		\begin{equation}\label{varbary1}
	VaR_{\alpha}(X^{+},\lambda)=\overline{m_{\lambda}}+\overline{\sigma_{\lambda}}G_{Z}^{-1}(\alpha),
	\end{equation}
where $Z=\frac{x_{q}-\overline{m_{\lambda}}}{\overline{\sigma_{\lambda}}}$, $G_{Z}(.)$ is the cumulative distribution functions of the standard random variable, $\overline{m_{\lambda}}=\sum_{j=1}^{N}\lambda_{j}m_{j}$,  $\overline{\sigma_{\lambda}}=\sum_{j=1}^{N}\lambda_{j}\sigma_{j},$ and $\lambda_{j}\geq 0$, $1\leq j\leq N$, $\sum_{j=1}^{N} \lambda_{j}=1.$
\end{theorem}
	
\begin{proof}
	It follows straightforwardly from  \ref{defq} and  \ref{propo}.
\end{proof}

The Wasserstein Barycenter Conditional Value-at-Risk is established next:

\begin{theorem}\label{tCV}
Let $\mu_{1},...,\mu_{N}$ be location and scale measures with corresponding means $m_1,...,m_N$  and  standard deviations $\sigma_1,...,\sigma_N$, then the Wasserstein Barycenter Conditional Value-at-Risk ($TCVaR_{\alpha}(X^{+},\lambda)$) is given by
\begin{equation}\label{varbary2}
TCVaR_{\alpha}(X^{+},\lambda)=\overline{m_{\lambda}}+\frac{\frac{1}{\overline{\sigma_{\lambda}}}g_{Z}(G_{Z}^{-1}(\alpha))}{1-\alpha}\overline{\sigma_{\lambda}}^{2}\sigma_{Z}^2.
\end{equation}
where $Z=\frac{x_{q}-\overline{m_{\lambda}}}{\overline{\sigma_{\lambda}}}$, $g_{Z}(.)$ and $G_{Z}(.)$ are the density and cumulative distribution functions of the standard random variable, $\overline{m_{\lambda}}=\sum_{j=1}^{N}\lambda_{j}m_{j}$,
$\overline{\sigma_{\lambda}}=\sum_{j=1}^{N}\lambda_{j}\sigma_{j},$and $\lambda_{j}\geq 0$, $1\leq j\leq N$, $\sum_{j=1}^{N} \lambda_{j}=1.$\\
\end{theorem}

\begin{proof}
Note that
$$TCVaR_{\alpha}(X^{+},\lambda)=\frac{1}{1-\alpha}\int_{VaR_{\alpha}(X^{+})}^{\infty}x.\frac{c}{\sigma}g\left(\frac{1}{2}\left(\frac{x-\overline{m_{\lambda}}}{\overline{\sigma_{\lambda}}}\right)^{2}\right)dx$$
and by letting  $Z=\frac{x-\overline{m_{\lambda}}}{\sigma}$, we have
$$TCVaR_{\alpha}(X^{+},\lambda)=\frac{1}{1-\alpha}\int_{Z_{q}}^{\infty}c(\overline{m_{\lambda}}+z\overline{\sigma_{\lambda}})g\left(\frac{1}{2}z^{2}\right)dz$$
$$=\overline{m_{\lambda}}+\overline{\sigma_{\lambda}}\frac{1}{1-\alpha}\int_{Z_{q}}^{\infty}cz.g\left(\frac{1}{2}z^{2}\right)dz$$
$$=\overline{m_{\lambda}}+\frac{\frac{1}{\overline{\sigma_{\lambda}}}g_{Z}(G_{Z}^{-1}(\alpha))}{1-\alpha}\overline{\sigma_{\lambda}}^{2}\sigma_{Z}^2.$$
\end{proof}

We now illustrate Theorems \ref{tV} and \ref{tCV} under the Gaussian distribution. \\

\textbf{Normal Distribution:} Let $\mu_{1},...,\mu_{N}$ be Normal measures with corresponding means $m_1,...,m_N$  and  standard deviations $\sigma_1,...,\sigma_N$; then the $VaR_{\alpha}(X^{+},\lambda)$ and the $TCVaR_{\alpha}(X^{+},\lambda)$ are given by

\begin{equation}\label{varbary}
VaR_{\alpha}(X^{+},\lambda)=\overline{m_{\lambda}}+\overline{\sigma_{\lambda}}\Phi^{-1}(\alpha),
\end{equation}

\begin{equation}
TCVaR_{\alpha}(X^{+},\lambda)=\overline{m_{\lambda}}+\overline{\sigma_{\lambda}}\frac{\phi(\Phi^{-1}(\alpha))}{1-\alpha},
\end{equation}

here  $\phi$ stands for the standard Gaussian distribution and $\Phi^{-1}$ is the inverse of the standard Gaussian distribution, $\overline{m_{\lambda}}=\sum_{j=1}^{N}\lambda_{j}m_{j}$,  $\overline{\sigma_{\lambda}}=\sum_{j=1}^{N}\lambda_{j}\sigma_{j},$ and $\lambda_{j}\geq 0$, $1\leq j\leq N$, $\sum_{j=1}^{N} \lambda_{j}=1.$\\

\section{Empirical Analysis: Portfolio Risk under normal Model}

Randomness in financial markets has promoted important research about robust measures of market risk. This problem motivates a profuse study about market risk. An issue involving the risk of loss for an investment under multifactor movements in a market. Some dynamical risk factors consider the interest and exchange rates, commodity risks and capital, among others. Thus, this section focus on estimation of the aggregation VaR for a risk portfolio
ruled by  Nasdaq and S\&P500 stock indices. The Nasdaq log-returns and the S\&P500 log-returns will be denoted as  $X_1$ and $X_2$, respectively.  In this case, the portfolio log-return,  $X^+$, has the form $X^+=\lambda_{1} X_1+\lambda_{2} X_2$. Here $\lambda=(\lambda_{1}, \lambda_{2})$ and $\lambda_{1}$ and  $\lambda_{2}$ are the portfolio weighs of the assets 1 and 2, $X_1$ and $X_2$. Without loss of generality, a portfolio under equal weights, in both indices, is considered. However, it is not a strict restriction and they can change freely. Finally, for the marginal returns a normal distribution is proposed and a one-day period VaR will be considered.

\subsection{Results of Wasserstein Barycenter approach}

Next, the Wasserstein Barycenter VaR is computed by using \ref{varbary}. In this case, each stock is ruled by Normal distribution. We follow the 2972 daily closing prices given by \cite{Landsman(2003)}; a database ranged from January 2nd, 1992 to October 1st, 2003. The dataset is divided into two parts: in sample period and test period. Sample period starts on January 3rd, 1995 and ends with December 7th, 2000. It consist of 750 daily returns of each stock index and offers the historical information needed for estimating VaR.  The test period starts on December 8th, 2000 to December 4th, 2004. VaR estimate accuracy is measured by using the Kupiec test for backtesting the method in small quantiles  $\alpha = 0.1,0.05,0.01,0.005.$ Software R is used for all computations. Table \ref{stati} shows the descriptive statistics of both series.

\begin{table}[H]
	\caption{Descriptive statistics for  log-returns series of daily  Nasdaq  and S\&P500  stock indices  }
	\begin{center}
\begin{tabular}{|c|c|c|c|c|c|c|}
		\hline
	Statistics & Nasdaq  & S\&P500   \\
	\hline

	Mean& 0.00038 & 0.00030  \\
	\hline
	Mean (annualized)& 10.141\% &7.857\%  \\
	\hline 	
	Standard Deviation &0.01694 &0.01076\\
	\hline 	
	Min.& -0.1016800   &  -0.0711275   \\
	\hline
	Median& 0.00122   & 0.00028  \\
	\hline 	
	Max.& 0.13255 & 0.05574  \\
	\hline
	Excess of Kurtosis& 4.91481 &3.78088\\
	\hline
	Asymmetry& 0.01490 &-0.10267\\
	\hline
\end{tabular}
\end{center}
\label{stati}
\end{table}

 According to Table \ref{stati}, the return series distributions of Nasdaq  and S\&P500 have small asymmetry, but strong kurtosis, in particular the first one. Note also that both series present positive means (annualized). \\

Thus, the results of Section \ref{principal} can be used for estimation of the the aggregation Value-at-Risk by using the equation

$$ VaR_{\alpha}(X^{+},\lambda)=-\overline{m_{\lambda}}-\overline{\sigma_{\lambda}}\Phi^{-1}(\alpha),$$

where  $\Phi^{-1}$ holds for the inverse of the standard Gaussian distribution, and the mean and the standard deviation are computed via $\overline{m_{\lambda}}=\sum_{j=1}^{N}\lambda_{j}m_{j}$, and $\overline{\sigma_{\lambda}}=\sum_{j=1}^{N}\lambda_{j}\sigma_{j},$ respectively. \\

The model computation includes both "unfiltered" and "filtered" forms. The filtered model case has also considered the volatility changes of the instrument. Such model will be referred as  Normal*. In the unfiltered Normal VaR (Normal), all the  $\sigma_{j}$'s, $j=1,,,N$ receive the same value of the sample standard deviation. But in the filtered Normal VaR (Normal*) the $\sigma_{j}$'s are estimated with a  Exponentially Weighted Moving Average (Ewma) model, where $\sigma_{t}=\sqrt{(1-\zeta)x_{i}^2+\zeta \sigma_{t-1}^2}$.

Kupiec test evaluates the performance by computing  the exceptions number in the corresponding test period.  In this case, $H_0:$ represents the null hypothesis and $1-\alpha$ is the probability of an exception occurrence. If $m$ is the number of observations for the test period and $x$ denotes the expected frequency of exceptions, then $h=\frac{x}{m}$ is the difference between the observed frequency of losses and VaR. The test statistics,

The corresponding test statistics is given by,

$$LR=-2[ln(p^x(1-p^{m-x}))-ln(h^x(1-h)^{m-x})]\sim \chi^2(1)$$.

It rejects the null hypothesis, with a 95\% confidence level, for $LR>\chi^2(1)$. In that case, the VaR estimations are not statistical meaningful generated by the particular VaR model, see \cite{McNeil(1999)}.\\

The dataset  under consideration is now divided into sample and test periods, with a selected window of 750 observations; and since there are 2971 observations available, then 2220 VaR tests can be performed at each level can be performed. The corresponding results are presented in Table  \ref{bary}.

\begin{table}[H]
	\caption{Wasserstein Barycenter Value-at-Risk, for t = 751 to 2971, number of exceptions (in brackets) where the estimated VaR was exceeded by the portfolio loss with $\alpha = 0.01, 0.05, 0.01, 0.005.$ P-values of tests.}
	\begin{center}
		\begin{tabular}{|c|c|c|c|c|c|c|}
			\hline
				Model &0.1 (222) &0.05 (111)&0.01 (22) &0.005 (11)		\\
			\hline
			\hline
					Wasserstein Barycenter-N &  0.0489  &0.0442 &  0.0312  &0.0243  \\		
			
			\hline
			Number of exceptions &225 &130  &46 &30
			\\
			\hline 	

			P-Value &0.8323& 0.0713& 9.1307e-06 &2.7020e-06
			\\
			\hline

			 VaR Model rejected &No  &No   & Yes    &Yes
				\\
			\hline
			\hline
			
			Wasserstein Barycenter-N* &  0.0387  &0.0349&  0.0247 &0.0192 \\
			\hline

			Number of exceptions &207 & 110 &23 &16
			\\
			\hline 	
			P-Value & 0.2837&  0.9223& 0.8653& 0.1668
			\\
			\hline
			VaR Model rejected &No &No &No &No
			\\
			
			\hline
			\hline
				
		\end{tabular}
	\end{center}
	\label{bary}
\end{table}

For all the $\alpha$ levels, the Wasserstein Barycenter-N* model showed the best performance in the VaR estimation. Moreover,  for $\alpha= 0.1,0.05$,  the Wasserstein Barycenter-N model also provided a satisfactory behavior.
In terms of the Kupiec test, applied to the number of exceptions for the Wasserstein Barycenter-N * model,  the null hypothesis were not rejected for all the $\alpha$ levels under consideration. In particular, high p values of  0.2837, 0.9223, 0.8653 and 0.1668  were obtained for Wasserstein Barycenter-N * model, and  0.8323 and 0.07130 for the Wasserstein Barycenter-N model at  $\alpha= 0.1,0.05$ levels.

\subsection{Comparisons }

According to \cite{Jianping(2012)}, most of the banks consider simple methods for assessing risk.  Complex methods involving copulas, simulation, hybrids models and advanced probability theory are implemented by few large banks.

Next apart will provide a summary of the variance-covariance and simple summation methods. Then a performance comparison with the proposed model is given. At the end, the new method is also confronted with the sophisticated hybrid copulas methodology given by \cite{Landsman(2003)}.

\subsubsection{Basic Multivariate VaR Models}

In the referred context of \cite{Jianping(2012)}, a research by the IFRI and CRO Forum's24 showed that 60\% or more of the studied banks consider simple approaches as the var-cov method in order to aggregate risk. And at least a little more advanced approaches (such as the supported methods by simulation) were used by only 20\% or less of the financial institutions in the survey. In fact, the summation setting prevails in most of the banks, a general model including variance-covariance method and simple summation. In the latest case the risks can be seen as explicitly weighted, meanwhile for the simple integration, the risks appear implicitly weighted.

This section computes the total VaR aggregated under several methods and  different confidence levels. For correctness of our results different performances with published results were studied. The methodology starts with a series of univariate portfolio returns; then by the use of simple sum and variance-covariance methods, the  aggregate VaR is estimated. These results are briefly described in the next lines.\\

\textbf{Simple Summation method}: The integration of $N$ risks is intuitively reached by aggregating the risks under summation of the particular $VaR_{\alpha}(X_{i})$ of each risk $X_{i}$, $i=1,\ldots,N$. Then the total aggregated VaR, $VaR_{\alpha}(X^{+})$, is expressed as:

\begin{equation}\label{sum1}
VaR_{\alpha}(X^{+})=-\sum_{i=1}^{N} VaR_{\alpha}(X_{i}),
\end{equation}
 see \cite{Embrecht(2013)} and \cite{Li(2015)}.\\

Now, as usual, the Gaussian model supports several methods in probability and statistics studies. In particular, the VaR of multivariate Gaussian laws are a common parametric method for multivariate VaR models. The technique supposes a multivariate Gaussian distribution (with mean  $\mu$ and covariance $\Sigma$ ) for  the returns of  the components in the portfolio. The method is characterized as follows:

\textbf{Variance-Covariance method}: If $\sigma_{+}=\sqrt{\lambda \Sigma \lambda'}.$ and  $\mu_{+}=\lambda \mu$ are respectively the deviation and expected portfolio return, then, the estimation of the Value at Risk for the corresponding multivariate Gaussian distribution returns is given by

\begin{equation}\label{cov}
VaR_{\alpha}(X^{+})=-\mu_{+}-\sigma_{+}\Phi^{-1}(\alpha),
\end{equation}

as before,  $\Phi^{-1}$ represents the inverse of the standard Gaussian distribution.

The covariance matrix and the mean vector in the Var-Cov approach are frequently unknown, then the model requires extra estimates taken from the current observations, see \cite{Jianping(2012)} and  \cite{Li(2015)}.\\

\begin{table}[H]
	\caption{Value-at-Risk for t = 751 to 2971, number of exceptions (in brackets) where the estimated VaR was exceeded by the portfolio loss with $\alpha = 0.01, 0.05, 0.01, 0.005.$}
	\begin{center}
		\begin{tabular}{|c|c|c|c|c|c|c|}
			\hline
		Model &0.1 (222) &0.05 (111)&0.01 (22) &0.005 (11)		\\
			\hline
			Simple summation& 0.0978 (30)& 0.0884 (14) & 0.0625 (4)& 0.0487(3)\\
			\hline
								
		Wasserstein Barycenter-N* & 0.0387 \textbf{(207)} &0.0349 \textbf{(110)}&  0.0247 \textbf{(23)} &0.0192 \textbf{(16)}
		\\
		\hline 	
		Wasserstein Barycenter-N &  0.04891 (225) &0.0442 (130)&  0.03123 (46) &0.0243 (30)\\
		\hline 	
		Var-Covar& 0.0477  (248)& 0.0430  (145)& 0.0304 (52)&  0.0237(38)	\\
		\hline
	\end{tabular}
	\end{center}
	\label{sum1}
\end{table}

A summary of the results are given next:\\
\begin{itemize}
	\item According Kupiec test, the basic multivariate VaR approaches does not predict future losses properly. Exception number in the test period is small under the Simple Summation; a fact explaining future loss overestimation. In contrast, the exception number under the variance-covariance method is large; promoting an underestimation of the future losses.

	\item The  Wasserstein Barycenter-N* and  Wasserstein Barycenter-N model exhibits a remarkable performance for future losses predictions. Moreover, Wasserstein Barycenter methods are stronger respect the other VaR models, because in the same reference time they provide a small exception probability and then a high-level capital reserve is not required.
	
In the set of the analyzed methods, the VaR Forecasting at all confidence levels were achieved in high performance by the proposed Wasserstein Barycenter-N* model. In fact, the Wasserstein Barycenter-N exhibited a better VaR forecasting than Var-Covar and Simple Summation methods in all the confidence levels. Empirical results also demonstrated the known fact that the Simple Approach provides an upper bound for the true VaR. In particular, for a confidence level   of 0.1\%, a VaR of 0.0387 derived by Wasserstein Barycenter-N* was the third part of the value 0.0978  based on the Simple Summation method. In such comparison context, our approach provides a number of possibilities for a wide bank class. Thus, under a conservative  Wasserstein Barycenter VaR compared with the general average, the proposed method, indexed by different types of weights, allow several criteria in order to improve the results. In contrast, the Var-covar method is preferably optimistic. Finally, the notable closed property for the barycenter in the class of location-scale distribution, just require a profuse knowledge of the prior barycenter under the selected distribution. Then, in the Gaussian case,  an exact formula for the value at risk  Wasserstein Barycenter can be derived and applied into the complete reference class of distributions. This opens an interesting perspective for manage risk studies of series based on non Gaussian models, because a complete mathematical description of the Wasserstein Barycenter VaR can be found under the selected prior distribution.

\end{itemize}

\subsubsection{Others Aggregate VaR models}

Robust multivariate methods including Copulas and ARCH models also involve VaR estimation. However, under a big number of assets,  such models produces biased parameter estimations and demand a high computational cost. For completeness we contrast our methods with  those derived by \cite{Jianping(2012)}.
\begin{table}[H]
	\caption{Value-at-Risk for t = 751 to 2971, number of exceptions (in brackets) where the estimated VaR was exceeded by the portfolio loss with $\alpha = 0.01, 0.05, 0.01, 0.005.$}
	\begin{center}
		\begin{tabular}{|c|c|c|c|c|c|c|}
			\hline
		Model & 0.1 (222) &0.05 (111)&0.01 (22) &0.005 (11)	\\
			\hline
					Wasserstein Barycenter-N* & 0.0387 \textbf{(207)} &0.0349 \textbf{(110)}&  0.0247 \textbf{(23)} &0.0192 (16)
			\\
			\hline 	
			Wasserstein Barycenter-N &  0.0489 (225) &0.0442 (130)&  0.0312 (46) &0.0243 (30)
			\\
			\hline
			SJC Copula + GARCH-E&& 0.0558 (124) &0.0104 (23) &0.0041 \textbf{(9)}\\
			\hline
			Bivariate GARCH (BEKK)& & 0.0819 (182)& 0.0338 (75)& 0.0248 (55)	\\
			\hline
			Bivariate GARCH (DCC) &&0.0432 (96)& 0.0140 (31)& 0.0113 (25)	\\
			\hline
			EWMA (Bivariate) &&0.0387 (86)& 0.0144 (32) &0.0104 (23)	\\
			\hline
			GARCH-N (Portfolio) && 0.0666 (148) & 0.0207 (46) &0.0144 (32)\\
			\hline
			GARCH-t (Portfolio) &&0.0693 (154)& 0.0131 (29)& 0.0104 (23)\\
			\hline
			EWMA (Portfolio) &&0.0527 (117) &0.0135 (30)& 0.0099 (22)\\
			\hline
			H.S. (Portfolio)&& 0.1220 (271)& 0.0293 (65) &0.0144 (32)\\
			\hline
		\end{tabular}
	\end{center}
	\label{sum}
\end{table}

Note that the Wasserstein Barycenter-N* model gives a satisfactory VaR estimation for $\alpha=0.1$, $\alpha=0.05$ and  $\alpha = 0.01$. When $\alpha = 0.01$ the estimation equals the result of the hibrid SJC Copula + GARCH-E method; for $\alpha=0.005$ it is near to the hibrid SJC Copula + GARCH-E model, which was proposed by \cite{Jianping(2012)} as the best approach. In fact, note that SJC Copula + GARCH-E method succeeds in the VaR forecasting at a  99.5\% confidence level, but it fails for 90\% and 95\%. Now, kupiec test given in Table \ref{bary} highlights that our methods show a high performance in all the  confidence levels. Moreover, inclusion of GARCH models should improved our methods, but for space reasons, we will leave that study for a future work. Given the robustness of GARCH models, the VaR estimation will require a few violentios and then a minor capital reserve should be demand.\\

 \subsection{COVID-19: 2020 Stock Market Crashs}

We end this section by showing the behavior of our method in a crucial modern problem. Explicitly, we will research the impact of COVID-19 in the 2020 Stock Market Crashs measured in the Wall Street indexes of  NASDAQ Composite, S\&P 500 and Dow Jones Industrial Average. These indexes reported historical loss levels, only registered in the financial crisis of 2008.\\

As usual, the complete data set is split in sample and test periods. Sample data ranges from January 4th, 2010 to December 31st,
2018. The 2264 daily return for every each stock index also refers as the required historical
data for a plausible VaR estimation. The test data, for detecting the performance of the value at risk models, ranges from January 2nd,
2019 to July 31st, 2020. The  VaR is estimated for each day in the test period (399 days), with the information offered by the 2663 observations ahead of it. Finally, the performance of the VaR model is measured by a comparison of the current loss and the estimated value at risk. The division for the sample and test periods are summarized in Table \ref{period}.\\

\begin{table}
	\caption{Sample and test periods}
	\begin{center}
		\begin{tabular}{|c|c|c|c|c|c|c|}
			\hline
			Period& In-sample period &test period &Total\\
			\hline
			Date& 4/1/2010-31/12/2018& 2/1/2019-31/7/2020&
			\\
			\hline
			Number of observations &2264 &399& 2663
			\\
			\hline 	
		\end{tabular}
	\end{center}
	\label{period}
\end{table}

Basic descriptive statistics analysis for the interval 2010-2019 and the 2020 year are given in Tables \ref{stati2} and \ref{stati3}, respectively.

\begin{table}
	\caption{Descriptive statistics for daily log-returns of Nasdaq, S\&P500 and Dow Jones during 2010-2019.}
	\begin{center}
		\begin{tabular}{|c|c|c|c|c|c|c|}
			\hline
			Statistics & Dow Jones 	& S\&P500&	Nasdaq Composite  \\
			\hline
			
			Mean& 0.00040
			& 0.00042& 0.00055
			\\
			\hline
			
			Median&   0.00058
			& 0.00060 & 0.00093
			\\
			\hline 	
			Standard Deviation &0.00887
			&0.00932& 0.01076 			
			\\
			\hline 	
			Min.& -0.057061
			&  -0.06896
			&-0.07149
			\\
			\hline
			Max.& 0.04864
			& 0.04840
			& 0.05672
			\\
			\hline
			Excess of Kurtosis& 4.04009
			&4.59922 &3.43256
			\\
			\hline
			Asymmetry& -0.47697
			&-0.49677 &
			-0.45537
			\\
			\hline
		\end{tabular}
	\end{center}
	\label{stati2}
\end{table}

\begin{table}
	\caption{Descriptive statistics for daily log-returns of Nasdaq, S\&P500 and Dow Jones in 2020}
	\begin{center}
		\begin{tabular}{|c|c|c|c|c|c|c|}
			\hline
			Statistics & Dow Jones 	& S\&P500&	Nasdaq Composite  \\
			\hline
			
			Mean&  -0,00052 & 0,00008 	& 0,00123
			
			\\
			\hline
			
			Median&  0,00109 & 0,00284  & 0,00427

			\\
			\hline 	
			Standard Deviation &0,02929		 &0,02718		
			&0,02703
			
			\\
			\hline
			Min.& -0,13841
			& -0.12765 	&-0.13149
			
			\\
			\hline
			Max.& 0.10764 	& 0.08968 	& 0.08935
			
			\\
			\hline
			Excess of Kurtosis& 5,52874		&5,40858	
			&5,62956
			
			\\
			\hline
			Asymmetry& -0,64012			&-0,69653		&
			-0,90154
			
			\\
			\hline
		\end{tabular}
	\end{center}
	\label{stati3}
\end{table}

Note that 2020 exhibits extremal high volatilities, in fact all the indexes highlight a notorious difference between that year and the interval 2010-2019. In addition, the minimum returns of the three indicators were reached on March 16th, 2020.

Next we compare the forecasting performance of the Wasserstein Barycenter-N* method with the basic multivariate VaR models. The best challenge for both methods resides in the financial market behavior of the 2020 year and its high volatility. First, Figure \ref{fig:grafica1} shows the forecast of the next trading day VaR for 2019 and 2020 by using both methods. Under a one day holding period, the models were computed for a 99\% confidence level. 

\begin{figure}
	\centering
	\includegraphics[width=0.9\linewidth]{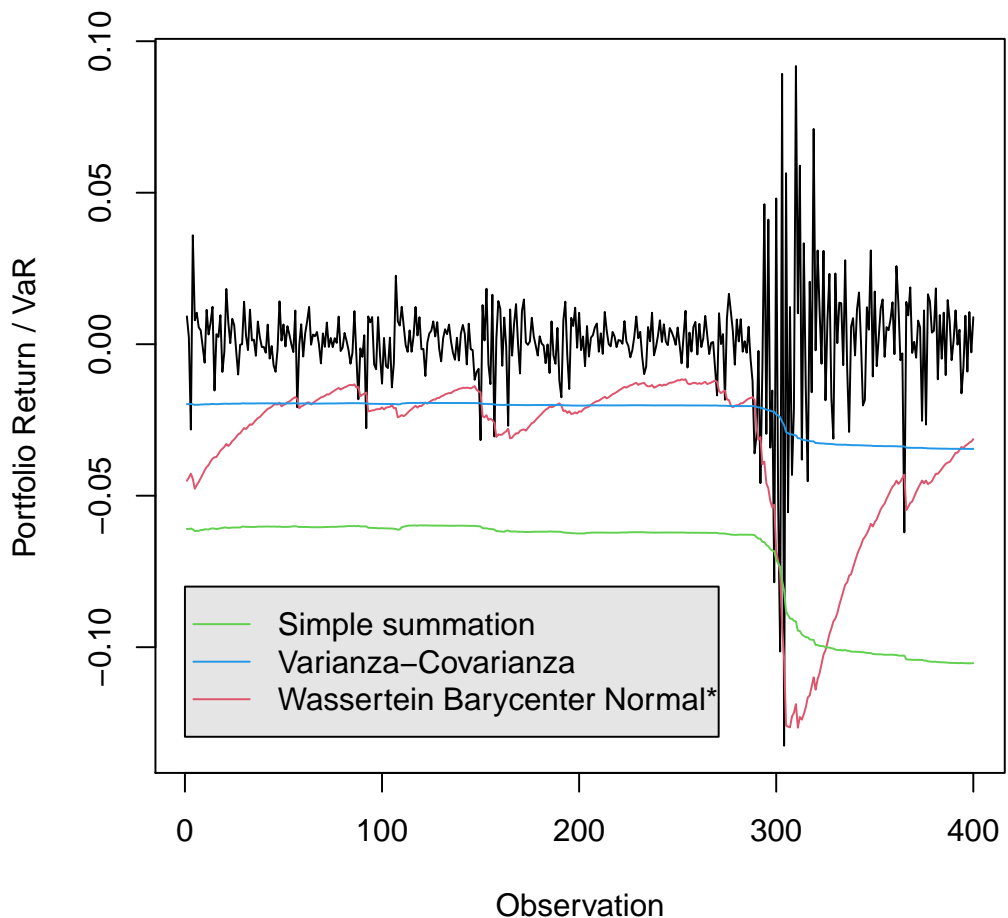}
	\caption{Daily Return series and VaR estimates}
	\label{fig:grafica1}
\end{figure}
As it can be seen the basic multivariate VaR models provide poor estimates for future losses. The VaR estimates rises improperly VaR estimates and unacceptable test period exceptions. However, the VaR estimation given by the Wasserstein Barycenter-N* VaR model is highly accurate. The new approach fits the volatile movements of the returns and predicts future losses notably, in comparison with the basic multivariate VaR approach. Moreover, according to  Figure \ref{fig:grafica1} the variance-covariance approach cannot follow the strong volatility since February 2020; and always presents underestimation. In terms of Simple Summation model, the test period shows lower violations but a strong conservative characteristic is noted. Finally, our proposed methods reach satisfactory VaR forecast in normal periods and extremal periods for  high volatility of the periods perform  The results indicate that the proposed approach provides satisfactory forecasts of VaR not only for the "normal periods" but also for the periods of high volatility due to the COVID-19 stock market influence.

\section{Concluding remarks}

This work has proposed a new multivariate risk measure model  based on the Wasserstein barycenters of probability measures under a location and scatter class. The method was compared with basic and sophisticated models under a portfolio characterized by S\&P500 and Nasdaq stock indices. The new model was also compared in United States market indices of high volatility during the current COVID-19. Kupiec test was used for assessing the performance of the existing and new methods.

The new approach is based on a notable property: Wasserstein barycenters of measures supported on location and scatter family belong to the same class. Then the paper proposed exact formulae for the Wasserstein Barycenter Value at Risk and the Wasserstein Barycenter Conditional Value at Risk for the addressed location and scatter family. This promotes a new setting for building robust risk measures and aggregation of different VaR in multiple financial markets. The closed form formulae are easily programmed for applications.

The technique can be used for non Gaussian distributions, opening an interesting future research for more robust risk measures.\\

\section*{acknowledgements}
	
	The authors wish to express their gratitude to the Doctorate in Mathematics (Doctoral School of Mathematics, IT and Telecommunications, University of Toulouse, Toulouse, France) and to the Doctorate in Modelling and Scientific Computing (Faculty of Basic Sciences, University of Medellin, Medellin, Colombia).

%
\section*{Conflict of interest}
The authors declare that they have no conflict of interest.


\begin{thebibliography}{9}
	
\bibitem[Agueh and Carlier (2011)]{Agueh(2011)} M. Agueh and G. Carlier. (2011) Barycenters in the Wasserstein space. SIAM J. Math. Anal., 43(2). 904--924.

\bibitem[Alvarez et al. (2018)]{Alvarez(2018)} Alvarez-Esteban, P. del Barrio, E., Cuesta-Albertos, J., \& Matrán, C. (2018). Wide consensus aggregation in the Wasserstein space. Application to location-scatter families. Bernoulli. 24. 3147-3179.

\bibitem[Barone et al.(2015)]{Barone(2015)} Barone-Adesi, Giovanni \& Giannopoulos, Kostas \& Vosper, Les. (2015). Estimating the Joint Tail Risk Under the Filtered Historical Simulation. An Application to the CCP's Default and Waterfall Fund. SSRN Electronic Journal.

\bibitem[Bigot et al. (2018)]{Bigot(2018)} Bigot J, Gouet R, Klein T, López A. (2018). Upper and lower risk bounds for estimating the Wasserstein barycenter of random measures on the real line. Electron. J. Stat. 12: 2253--2289.

\bibitem[Boissard et al. (2015)]{Loubes(2015)} Emmanuel Boissard, Thibaut Le Gouic, Jean-Michel Loubes. (2015) Distributions template estimate with Wasserstein metrics. Bernoulli, 21(2): 740--759.

\bibitem[Cédric Villani. (2008)]{Villani(2008)} Cédric Villani. (2008). Optimal transport: old and new, volume 338. Springer Science \& Business Media.

\bibitem[Embrechts et al. (2002)]{Embrecht(2002)} Embrechts P,  McNeil A. J. and D Straumann (2002), Correlation and dependence proprieties in risk management: proprieties and pitfalls, in M Demster ed Risk Management: Value at Risk and Beyond, Cambridge
University Press.


\bibitem[Embrechts et al. (2013)]{Embrecht(2013)} Embrechts P., Giovanni P., Rüschendorf L.(2013). Model uncertainty and VaR aggregation, Journal of Banking \& Finance, 37(8): 2750--2764.

 \bibitem[Feng Yu and Erik Schlogl (2018)]{Feng(2018)} Feng Yu and Erik Schlogl. (2018). Model Risk Measurement Under Wasserstein Distance. No 393. Research Paper Series, Quantitative Finance Research Centre. University of Technology. Sydney.

\bibitem[Forum (2010)]{Forum (2010)} Forum Joint (2010) Developments in modelling risk aggregation. Bank for International Settlements, Basel.

\bibitem[Landsman et al.(2003)]{Landsman(2003)} Landsman, Z. M., and E. A. Valdez (2003). Tail conditional expectations for elliptical distributions. North Amer. Actuar. J.
7(4): 55--71.

\bibitem[Le Gouic and Loubes (2017)]{Loubes(2017)} Le Gouic, T., Loubes, J.M. Existence and consistency of Wasserstein barycenters. Probab. Theor. Related Fields 168(34): 901--917

\bibitem[Li et al. (2015)]{Li(2015)} Li, J., Zhu, X., Lee, C.-F., Wu, D., Feng, J., and Shi, Y. (2015). On the aggregation of credit, market and operational risks. Review of Quantitative Finance and Accounting, 44(1): 161--189.

\bibitem[Jianping  et al. (2012)]{Jianping(2012)} Jianping  Li,  Jichuang  Feng ,  Xiaolei  Sun,   Minglu  Li.  (2012). Risk integration mechanisms and approaches in banking industry.  International Journal of Information Technology \& Decision Making. 11(06): 1183--1213.

\bibitem[Jin and Lehnert (2018)]{Jin(2018)} Jin, X. \& Lehnert, T. (2018). Large portfolio risk management and optimal portfolio allocation with dynamic elliptical copulas. Dependence Modeling, 6(1): 19-46.


\bibitem[Kiesel et al. (2016)]{Kiesel(2016)} Kiesel, R., Ruhlicke, R., Stahl, G., Zheng, J. (2016). The Wasserstein Metric and Robustness in Risk Management. Risks. 4, 32.

 \bibitem[McNeil (1999)]{McNeil(1999)} McNeil, A. J., (1999). Extreme Value Theory for risk managers, mimeo, ETH Zurich.

 \bibitem[McNeil and Embrechts (2015)]{McNeil(2015)} McNeil, A. J., Frey, R. and Embrechts, P., (2015). Quantitative risk management: Concepts, techniques and tools. Princeton, New Jersey: Princeton university press.


\bibitem[Wagalath and Zubelli (2018)]{Wagalath(2018)} Wagalath, L. and Zubelli, J. (2018) A liquidation risk adjustment for value at risk and expected shortfall. International Journal of Theoretical and Applied Finance, 21 (3): 1-21.

\end{thebibliography}
\end{document}